\documentclass[11pt,a4paper]{article}

\usepackage[a4paper,margin=1in]{geometry}
\usepackage[T1]{fontenc}
\usepackage[utf8]{inputenc}
\usepackage{lmodern}
\usepackage{amsmath,amssymb,amsthm,mathtools,bm}
\usepackage{graphicx}
\usepackage{booktabs}
\usepackage{hyperref}
\usepackage{tikz}
\usepackage{pgfplots}
\usepackage{caption}
\usepackage{enumitem}

\pgfplotsset{compat=1.18}
\hypersetup{
  colorlinks=true,
  linkcolor=blue,
  citecolor=blue,
  urlcolor=blue
}

\newcommand{\dd}{\mathrm{d}}
\newcommand{\ee}{\mathrm{e}}
\newcommand{\ii}{\mathrm{i}}
\newcommand{\Res}{\mathop{\mathrm{Res}}}

\newcommand{\cH}{\mathcal{H}}
\newcommand{\cM}{\mathcal{M}}
\newcommand{\cO}{\mathcal{O}}
\newcommand{\cW}{\mathcal{W}}

\newcommand{\SH}{S_{\mathrm{H}}}
\newcommand{\SA}{S_{A}}
\newcommand{\Ff}{F}
\newcommand{\diS}{d_{i}S}
\newcommand{\dtotS}{dS_{\mathrm{tot}}}
\newcommand{\dmS}{dS_{m}}

\newcommand{\TempH}{T_{h}}
\newcommand{\Ahor}{A_{h}}
\newcommand{\VH}{V_{h}}
\newcommand{\EH}{E_{h}}

\newcommand{\rA}{\tilde r_{A}}
\newcommand{\TA}{T_{A}}
\newcommand{\EA}{E_{A}}

\theoremstyle{definition}
\newtheorem{definition}{Definition}[section]
\newtheorem{remark}{Remark}[section]
\newtheorem{proposition}{Proposition}[section]

\title{Non-Equilibrium Physics of Thermodynamicized Black Holes}
\author{Wen-Xiang Chen\\Guangzhou City University of Technology\\School of Physics and Materials Science, Guangzhou University\\wxchen4277@qq.com}
\date{\today}

\begin{document}
\maketitle

\begin{abstract}
Motivated by the entropy-functional viewpoint of emergent gravity and by recent residue-based constructions of black-hole thermodynamics, we formulate a non-equilibrium framework for thermodynamicized black holes. The central idea is to combine three ingredients: an entropy-functional criterion that selects on-shell backgrounds, the Euclidean/contour representation of horizon temperature through simple-pole singularities of the lapse function, and a topological residue classification of multi-horizon configurations. Building on these ingredients, we introduce a quasi-stationary non-equilibrium partition functional in which irreversible entropy production appears as an additive contribution to the singular action. The resulting formalism preserves the equilibrium relations in the adiabatic limit while extending them to driven black-hole configurations with matter/charge/rotation fluxes. We then apply the construction to Kerr--Newman-type black holes in constant-curvature \(f(R)\) gravity and show that the equilibrium entropy is still weighted by \(f'(R_0)\), whereas non-equilibrium corrections enter through flux-induced deformations of the effective thermodynamic action. The outer and inner horizons contribute opposite topological orientations, so the non-extremal Kerr--Newman family remains in the class \(W=0\) unless a horizon bifurcation or merger changes the singularity structure. Finally, we provide several explicit function plots---including an equilibrium/non-equilibrium free-energy model, a Kerr--Newman temperature curve, and an entropy-production law.

\end{abstract}

\section{Introduction}

Black-hole thermodynamics provides one of the most suggestive bridges between gravitation, statistical physics, and quantum theory \cite{Bekenstein1973,Hawking1975}. The Bekenstein--Hawking entropy and Hawking temperature imply that horizons behave as thermodynamic objects rather than merely as causal boundaries \cite{Bekenstein1973,Hawking1975}. At the microscopic level, string-theoretic and near-horizon state counting arguments further support this thermodynamic interpretation \cite{StromingerVafa1996,Strominger1998}. In parallel, thermodynamic geometry and black-hole microstructure analyses have provided an additional phenomenological language for discussing interaction patterns and phase structure \cite{Ruppeiner1995,WeiLiuMann2019}. In the broader emergent-gravity program, gravitational field equations may be interpreted as macroscopic equations of state, and the horizon sector is often treated as the natural carrier of thermodynamic information \cite{Jacobson1995,PadmanabhanParanjape2007,Padmanabhan2010}. At the same time, Euclidean quantum gravity and complex-analysis techniques indicate that horizon thermodynamics can be encoded in the singularity structure of analytically continued partition functions \cite{GibbonsHawking1977,HawkingPage1983,York1986,BradenBrownWhitingYork1990}.

The two uploaded preprints motivating the present manuscript emphasize complementary aspects of this program \cite{ChenTopological2025,ChenSingular2026}. The first develops a topological complex-analysis framework for Kerr--Newman black holes in \(f(R)\) gravity, where microstructure data are mapped to singularities on a complexified partition function and classified by a winding-number-like topological index \cite{ChenTopological2025}. The second organizes gravitational thermodynamics into canonical and grand-canonical singular ensembles, selected by an entropy-functional equilibrium condition and evaluated through contour integrals around simple poles of the lapse function \cite{ChenSingular2026}. These ingredients naturally suggest a further step: a non-equilibrium extension in which driven black-hole configurations are described by a generalized singular action that contains both reversible and irreversible contributions.

The goal of this paper is therefore modest but precise. We do not claim a new microscopic theory of quantum gravity. Instead, we construct a mathematically explicit and phenomenologically useful non-equilibrium framework for thermodynamicized black holes by synthesizing the entropy-functional language, the residue-temperature formula, and the topological classification of multi-horizon states. The formalism is designed to satisfy the following requirements:
\begin{enumerate}[label=(\roman*)]
    \item it must reduce to the usual equilibrium relations in the quasi-static limit;
    \item it must preserve the contour/residue interpretation of horizon temperature;
    \item it must accommodate fluxes of energy, angular momentum, charge, and effective particle number;
    \item it must clarify when the topological index changes and when it remains protected.
\end{enumerate}

The relativistic interpretation of redshifted intensive variables also remains important throughout, especially when one introduces chemical-potential-type quantities and quasi-local thermodynamic sectors \cite{Tolman1930,Klein1949}. Meanwhile, the topological viewpoint adopted here is consistent with recent developments in black-hole thermodynamic topology \cite{YerraBhamidipati2022,ChenTopological2025}.

The paper is organized as follows. Section~2 reviews the entropy-functional selection rule, the residue formula for horizon temperature, and the \(f(R)\)-corrected Kerr--Newman background. Section~3 introduces a quasi-stationary non-equilibrium partition functional and derives the corresponding entropy-balance law. Section~4 applies the construction to Kerr--Newman-type black holes in constant-curvature \(f(R)\) gravity and discusses the topological index of outer/inner horizon branches. Section~5 presents explicit function plots that visualize the effective free energy, the Kerr--Newman temperature curve, and the irreversible entropy production. Sections~6 and 7 develop the transition from \(f(R)\) gravity to a gravity-thermodynamized model and then illustrate the formalism through static spherically symmetric and FRW examples. Section~8 summarizes the main conclusions and limitations.

\section{Thermodynamicized Gravity Background}

\subsection{Entropy-functional selection of on-shell backgrounds}

Following the entropy-functional viewpoint of emergent gravity \cite{Jacobson1995,PadmanabhanParanjape2007,Padmanabhan2010,ChenSingular2026}, we consider an auxiliary vector field \(\xi^\mu\) associated with local horizon generators and define the functional
\begin{equation}
S_{\xi}[g,\xi] = \int_{V} \nabla_{\mu}\xi_{\nu}\,\nabla^{\mu}\xi^{\nu}\,\sqrt{-g}\,\dd^{4}x.
\label{eq:Sxi}
\end{equation}
Varying \(\xi^\mu\) while holding the metric fixed yields
\begin{equation}
\delta S_{\xi} = -2\int_{V}(\Box\xi_{\nu})\,\delta\xi^{\nu}\,\sqrt{-g}\,\dd^{4}x,
\label{eq:deltaSxi}
\end{equation}
so the stationary condition is
\begin{equation}
\Box\xi_{\nu}=0.
\label{eq:boxxi}
\end{equation}
If \(\xi^\mu\) is identified with a local Killing generator, then the standard identity
\begin{equation}
\Box\xi_{\nu}=-R_{\nu\lambda}\xi^{\lambda}
\label{eq:killingid}
\end{equation}
implies
\begin{equation}
R_{\nu\lambda}\xi^{\lambda}=0.
\label{eq:Rnulambda}
\end{equation}
Contracting with \(\xi^{\nu}\) gives the local equilibrium condition
\begin{equation}
R_{\mu\nu}\xi^{\mu}\xi^{\nu}=0.
\label{eq:nullconstraint}
\end{equation}
Equation~\eqref{eq:nullconstraint} is not the full field equation; rather, it selects the stationary or quasi-stationary backgrounds on which the thermodynamic description is meaningful \cite{Jacobson1995,PadmanabhanParanjape2007,Padmanabhan2010}.

\subsection{Residue representation of the horizon temperature}

For a static Euclideanized line element of the form
\begin{equation}
ds_{E}^{2}=f(r)\dd\tau^{2}+\frac{\dd r^{2}}{f(r)}+r^{2}\dd\Omega_{2}^{2},
\label{eq:Euclideanmetric}
\end{equation}
assume that the lapse function has a simple zero at \(r=r_{h}\):
\begin{equation}
f(r)=f'(r_{h})(r-r_{h})+\mathcal{O}\big((r-r_{h})^{2}\big).
\label{eq:simplezero}
\end{equation}
Regularity of the \((\tau,r)\)-plane implies
\begin{equation}
\beta_{h}=\frac{4\pi}{f'(r_{h})}.
\label{eq:betafh}
\end{equation}
Since
\begin{equation}
\Res_{r=r_{h}}\frac{1}{f(r)}=\frac{1}{f'(r_{h})},
\label{eq:residuebasic}
\end{equation}
we obtain the contour-temperature formula
\begin{equation}
\beta_{h}=4\pi\,\Res_{r=r_{h}}\frac{1}{f(r)} = \frac{2}{\ii}\oint_{\Gamma_{h}}\frac{\dd r}{f(r)}.
\label{eq:contourtemperature}
\end{equation}
This is the key technical bridge between horizon geometry and thermodynamic data \cite{GibbonsHawking1977,York1986,HawkingPage1983,ChenSingular2026}.

\subsection{Constant-curvature \texorpdfstring{$f(R)$}{f(R)} gravity and Kerr--Newman geometry}

In \(f(R)\) gravity the action is
\begin{equation}
I = \frac{1}{16\pi G}\int \dd^{4}x\,\sqrt{-g}\,f(R),
\label{eq:fRaction}
\end{equation}
and the field equations are
\begin{equation}
f'(R)R_{\mu\nu}-\frac{1}{2}f(R)g_{\mu\nu}+\left(g_{\mu\nu}\Box-\nabla_{\mu}\nabla_{\nu}\right)f'(R)=0.
\label{eq:fRfield}
\end{equation}
For constant curvature \(R=R_{0}\), the vacuum condition becomes
\begin{equation}
f'(R_{0})R_{0}=2f(R_{0}).
\label{eq:constantcurvature}
\end{equation}
A Kerr--Newman-type solution in this sector is characterized by
\begin{equation}
\Delta(r)=r^{2}+a^{2}-\frac{2Mr}{f'(R_{0})}+\frac{Q^{2}}{f'(R_{0})}-\frac{R_{0}}{12}(r^{2}+a^{2})r^{2}.
\label{eq:DeltafR}
\end{equation}
When \(R_{0}=0\) and \(f'(R_{0})=1\), one recovers the standard asymptotically flat Kerr--Newman case \cite{ChenTopological2025}. The outer and inner horizon radii are determined by \(\Delta(r_{\pm})=0\). The horizon entropy is weighted by the effective coupling:
\begin{equation}
S_{+}=\frac{f'(R_{0})A_{+}}{4G},
\label{eq:fRentropy}
\end{equation}
with \(A_{+}=4\pi(r_{+}^{2}+a^{2})\), in agreement with the Wald entropy formula \cite{Wald1993}. For the asymptotically flat case,
\begin{equation}
T_{+}=\frac{r_{+}-r_{-}}{4\pi(r_{+}^{2}+a^{2})},
\label{eq:KNtemperature}
\end{equation}
whereas the formal inner-horizon temperature is
\begin{equation}
T_{-}=\frac{r_{-}-r_{+}}{4\pi(r_{-}^{2}+a^{2})}.
\label{eq:innerT}
\end{equation}
The sign of \(T_{-}\) already signals the instability of the inner branch \cite{ChenTopological2025}.

\section{Quasi-Stationary Non-Equilibrium Extension}

\subsection{Generalized singular action}

In equilibrium, the singular contribution to the thermodynamic action is encoded by the contour integral of the appropriate thermodynamic observable divided by the lapse function \cite{GibbonsHawking1977,York1986,ChenSingular2026}. To extend this idea to non-equilibrium but slowly driven configurations, we introduce an evolution parameter \(\lambda\) and define a generalized singular action
\begin{equation}
\mathcal{I}_{\mathrm{sing}}^{\mathrm{neq}}[\lambda]
= \frac{2}{\ii}\oint_{\Gamma_{h}(\lambda)}\frac{\mathcal{X}(r,\lambda)}{f(r,\lambda)}\,\dd r,
\label{eq:Isingneq}
\end{equation}
with
\begin{equation}
\mathcal{X}(r,\lambda)=E(\lambda)-\Omega_{h}(\lambda)J(\lambda)-\Phi_{h}(\lambda)Q(\lambda)-\mu_{h}(\lambda)N(\lambda)+\Pi_{\mathrm{eff}}(r,\lambda).
\label{eq:Xdef}
\end{equation}
Here \(\Pi_{\mathrm{eff}}\) represents the irreversible sector, interpreted as the local source of entropy production or dissipative work. If \(\Pi_{\mathrm{eff}}\to 0\), the equilibrium singular action is recovered.

When \(\mathcal{X}(r,\lambda)\) is regular at the horizon, Eq.~\eqref{eq:Isingneq} reduces to
\begin{equation}
\mathcal{I}_{\mathrm{sing}}^{\mathrm{neq}} = \beta_{h}\left(E-\Omega_{h}J-\Phi_{h}Q-\mu_{h}N+\Pi_{h}\right),
\label{eq:Isingneqsimple}
\end{equation}
where \(\Pi_{h}=\Pi_{\mathrm{eff}}(r_{h},\lambda)\) and \(\beta_{h}\) is still determined by the residue formula \eqref{eq:contourtemperature}.

\subsection{Non-equilibrium partition functional and entropy balance}

A natural quasi-stationary partition functional is therefore
\begin{equation}
\mathcal{Z}_{\mathrm{neq}}\sim \exp\left[S_{h}-\mathcal{I}_{\mathrm{sing}}^{\mathrm{neq}}\right].
\label{eq:Zneq}
\end{equation}
Equivalently,
\begin{equation}
\ln \mathcal{Z}_{\mathrm{neq}}
= S_{h}-\beta_{h}\left(E-\Omega_{h}J-\Phi_{h}Q-\mu_{h}N\right)-\beta_{h}\Pi_{h}.
\label{eq:lnZneq}
\end{equation}
To separate reversible and irreversible parts, we write the entropy balance law as
\begin{equation}
\frac{\dd S_{h}}{\dd\lambda}
= \frac{1}{T_{h}}\left(\frac{\dd E}{\dd\lambda}-\Omega_{h}\frac{\dd J}{\dd\lambda}-\Phi_{h}\frac{\dd Q}{\dd\lambda}-\mu_{h}\frac{\dd N}{\dd\lambda}\right)+\dot{S}_{\mathrm{irr}},
\label{eq:entropybalance}
\end{equation}
with
\begin{equation}
\dot{S}_{\mathrm{irr}}\equiv \frac{\Pi_{h}}{T_{h}}\geq 0.
\label{eq:Sirrdef}
\end{equation}
Equation~\eqref{eq:entropybalance} is the non-equilibrium analogue of the first law. The first term is reversible, whereas \(\dot S_{\mathrm{irr}}\) captures dissipation, horizon viscosity, flux-induced mixing, or other irreversible channels.

A convenient integrated form is
\begin{equation}
\Delta S_{h}=\int_{\lambda_{i}}^{\lambda_{f}}\frac{1}{T_{h}}\left(\dd E-\Omega_{h}\dd J-\Phi_{h}\dd Q-\mu_{h}\dd N\right)+\int_{\lambda_{i}}^{\lambda_{f}}\dot{S}_{\mathrm{irr}}\,\dd\lambda.
\label{eq:DeltaS}
\end{equation}
This makes clear that the non-equilibrium framework is consistent with the second law whenever \(\dot S_{\mathrm{irr}}\ge0\).

\subsection{Adiabatic limit and canonical/grand-canonical sectors}

The formalism collapses to the equilibrium sectors discussed in the uploaded singular-ensemble paper when the driving vanishes \cite{ChenSingular2026}. In the canonical sector,
\begin{equation}
\mathcal{I}_{\mathrm{sing}}^{(A)}=\beta_{h}E,
\label{eq:canonicalequilibrium}
\end{equation}
whereas in the grand-canonical sector,
\begin{equation}
\mathcal{I}_{\mathrm{sing}}^{(B)}=\beta_{h}(E-\mu_{h}N).
\label{eq:grandcanonicalequilibrium}
\end{equation}
These sector choices are also naturally related to the Euclidean canonical and grand-canonical black-hole ensembles studied in the literature \cite{York1986,BradenBrownWhitingYork1990}. The rotating/charged black-hole extension simply augments the intensive variables to \((T_{h},\Omega_{h},\Phi_{h},\mu_{h})\). Thus the non-equilibrium proposal can be viewed as a flux-dressed version of the same residue calculus rather than as a qualitatively unrelated construction. The interpretation of local temperature and chemical potential in a gravitational field is also consistent with the classical relativistic equilibrium arguments of Tolman and Klein \cite{Tolman1930,Klein1949}.

\section{Application to Kerr--Newman-Type Black Holes in \texorpdfstring{$f(R)$}{f(R)} Gravity}

\subsection{Effective horizon thermodynamics}

For the asymptotically flat Kerr--Newman family, the horizon radii satisfy
\begin{equation}
r_{+}+r_{-}=2M,
\qquad
r_{+}r_{-}=a^{2}+Q^{2}.
\label{eq:rplusrminus}
\end{equation}
Hence the outer-horizon temperature can be written as a single-variable function,
\begin{equation}
T_{+}(r_{+};a,Q)=\frac{r_{+}-\dfrac{a^{2}+Q^{2}}{r_{+}}}{4\pi(r_{+}^{2}+a^{2})}.
\label{eq:Tplusonevar}
\end{equation}
The corresponding equilibrium entropy is
\begin{equation}
S_{+}(r_{+}) = \frac{\pi f'(R_{0})}{G}\left(r_{+}^{2}+a^{2}\right),
\label{eq:Splusr}
\end{equation}
again reflecting the Wald entropy weighting in \(f(R)\) gravity \cite{Wald1993,ChenTopological2025}. To model quasi-stationary non-equilibrium driving, we let the effective temperature be deformed as
\begin{equation}
T_{\mathrm{eff}}(r_{+};a,Q,\varepsilon)
= T_{+}(r_{+};a,Q)\left(1-\varepsilon \ee^{-r_{+}}\right),
\qquad 0\le \varepsilon<1,
\label{eq:Teff}
\end{equation}
where \(\varepsilon\) is a dimensionless non-equilibrium response parameter. Equation~\eqref{eq:Teff} should be understood as a phenomenological interpolation rather than an exact solution of the field equations; its purpose is to visualize how dissipative dressing lowers the effective horizon temperature while preserving the equilibrium limit \(\varepsilon\to0\).

The generalized non-equilibrium Gibbs functional is written as
\begin{equation}
\mathcal{G}_{\mathrm{neq}} = E-T_{\mathrm{eff}}S_{+}-\Omega_{+}J-\Phi_{+}Q+\mathcal{D},
\label{eq:Gneq}
\end{equation}
where \(\mathcal{D}\) denotes an explicitly dissipative correction. For an illustrative one-parameter model we choose
\begin{equation}
\widetilde{\mathcal{F}}_{\mathrm{eq}}(x)=\frac{1}{2}x-\pi\tau x^{2},
\qquad
\widetilde{\mathcal{F}}_{\mathrm{neq}}(x)=\frac{1}{2}x-\pi\tau x^{2}+\epsilon(x-x_{c})^{2},
\label{eq:Fmodels}
\end{equation}
with dimensionless horizon scale \(x\), control parameter \(\tau\), dissipative amplitude \(\epsilon\), and reference scale \(x_{c}\). This choice extends the equilibrium illustrative potential used in the uploaded singular-ensemble article \cite{ChenSingular2026}.

\subsection{Topological classification of horizon branches}

The topological interpretation of the residue structure follows from associating an orientation sign to each thermodynamic branch \cite{YerraBhamidipati2022,ChenTopological2025}. We define
\begin{equation}
W = \sum_{h}\sigma_{h},
\qquad
\sigma_{h}=\pm 1,
\label{eq:Wdef}
\end{equation}
where stable outer-horizon branches contribute \(+1\) and unstable inner-horizon branches contribute \(-1\). For a non-extremal Kerr--Newman black hole,
\begin{equation}
W=\sigma_{+}+\sigma_{-}=+1+(-1)=0.
\label{eq:Wzero}
\end{equation}
For a single-horizon or effectively single-branch configuration,
\begin{equation}
W=+1.
\label{eq:Wone}
\end{equation}
The crucial point is that non-equilibrium driving does not by itself change \(W\). The index changes only if the singularity structure changes qualitatively, for example through horizon merger, horizon creation, or annihilation. Therefore, small dissipative deformations preserve the topological class:
\begin{equation}
\delta W = 0
\qquad \text{for quasi-stationary deformations without horizon bifurcation.}
\label{eq:dWzero}
\end{equation}
This result gives a topological explanation of why the Kerr--Newman family remains in the \(W=0\) class under mild \(f(R)\) modifications and under weak non-equilibrium driving \cite{YerraBhamidipati2022,ChenTopological2025}.

\subsection{Non-equilibrium entropy production near the horizon}

A simple quadratic constitutive ansatz for irreversible production is
\begin{equation}
\dot S_{\mathrm{irr}} = \gamma \mathcal{J}^{2},
\qquad \gamma>0,
\label{eq:quadraticSirr}
\end{equation}
where \(\mathcal{J}\) is a dimensionless flux variable summarizing matter influx, shear, charge transport, or a combination thereof. Substituting Eq.~\eqref{eq:quadraticSirr} into Eq.~\eqref{eq:entropybalance} yields
\begin{equation}
\frac{\dd S_{h}}{\dd\lambda}
= \frac{1}{T_{h}}\left(\frac{\dd E}{\dd\lambda}-\Omega_{h}\frac{\dd J}{\dd\lambda}-\Phi_{h}\frac{\dd Q}{\dd\lambda}-\mu_{h}\frac{\dd N}{\dd\lambda}\right)+\gamma \mathcal{J}^{2}.
\label{eq:entropybalancequadratic}
\end{equation}
The last term is strictly non-negative and therefore compatible with the generalized second law. In a more microscopic treatment, \(\gamma\) should be derived from transport coefficients of the horizon fluid or from a Kubo-type relation, but the quadratic form already suffices to visualize the departure from equilibrium.

\section{Illustrative Function Plots}

The following plots are included directly in the manuscript so that the source file is self-contained. They are analytic illustrations of the formalism, not numerical solutions of the full \(f(R)\) field equations.

\subsection{Equilibrium and non-equilibrium effective free energy}

Figure~\ref{fig:freeenergy} shows the equilibrium model \(\widetilde{\mathcal{F}}_{\mathrm{eq}}\) and its dissipative deformation \(\widetilde{\mathcal{F}}_{\mathrm{neq}}\).

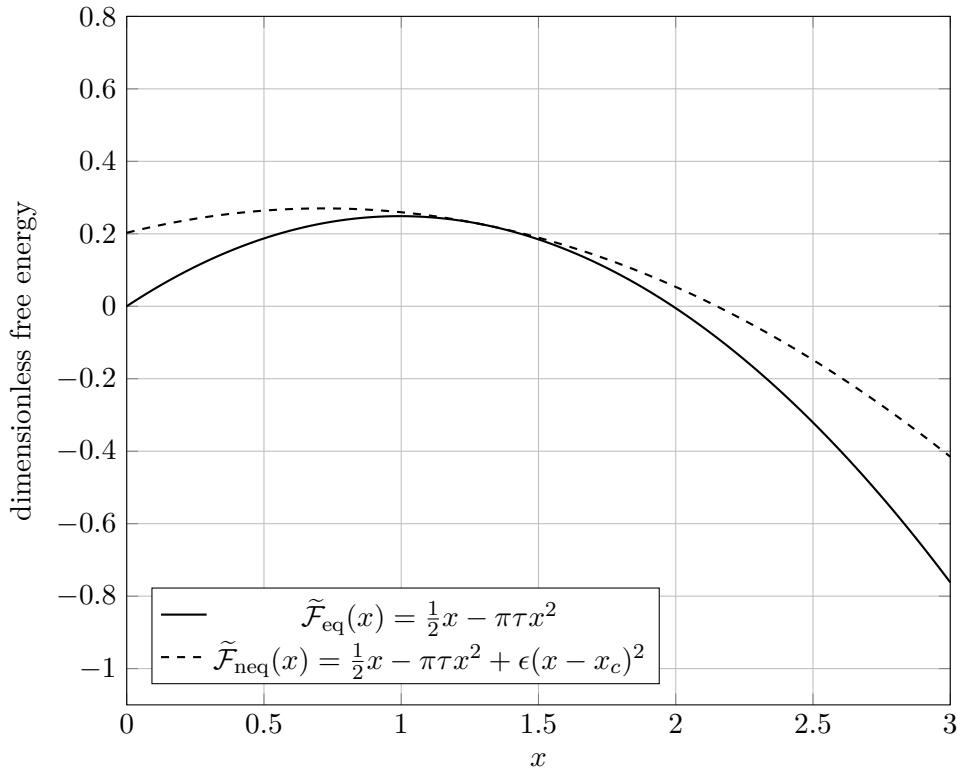
\begin{figure}[htbp]
\centering
\begin{tikzpicture}
\begin{axis}[
    width=0.78\textwidth,
    xlabel={$x$},
    ylabel={dimensionless free energy},
    xmin=0, xmax=3,
    ymin=-1.1, ymax=0.8,
    legend pos=south west,
    grid=both,
]
\addplot[smooth, thick, domain=0:3, samples=200] {0.5*x - pi*0.08*x^2};
\addlegendentry{$\widetilde{\mathcal{F}}_{\mathrm{eq}}(x)=\frac{1}{2}x-\pi\tau x^2$}
\addplot[smooth, thick, dashed, domain=0:3, samples=200] {0.5*x - pi*0.08*x^2 + 0.12*(x-1.3)^2};
\addlegendentry{$\widetilde{\mathcal{F}}_{\mathrm{neq}}(x)=\frac{1}{2}x-\pi\tau x^2+\epsilon(x-x_c)^2$}
\end{axis}
\end{tikzpicture}
\caption{Illustrative equilibrium and non-equilibrium effective free-energy branches for \(\tau=0.08\), \(\epsilon=0.12\), and \(x_c=1.3\). The dashed branch shows how dissipative driving lifts the effective potential relative to the equilibrium branch.}
\label{fig:freeenergy}
\end{figure}

\subsection{Kerr--Newman temperature curve}

Figure~\ref{fig:temperature} plots the equilibrium Kerr--Newman temperature from Eq.~\eqref{eq:Tplusonevar} together with the phenomenological non-equilibrium dressing \eqref{eq:Teff}.

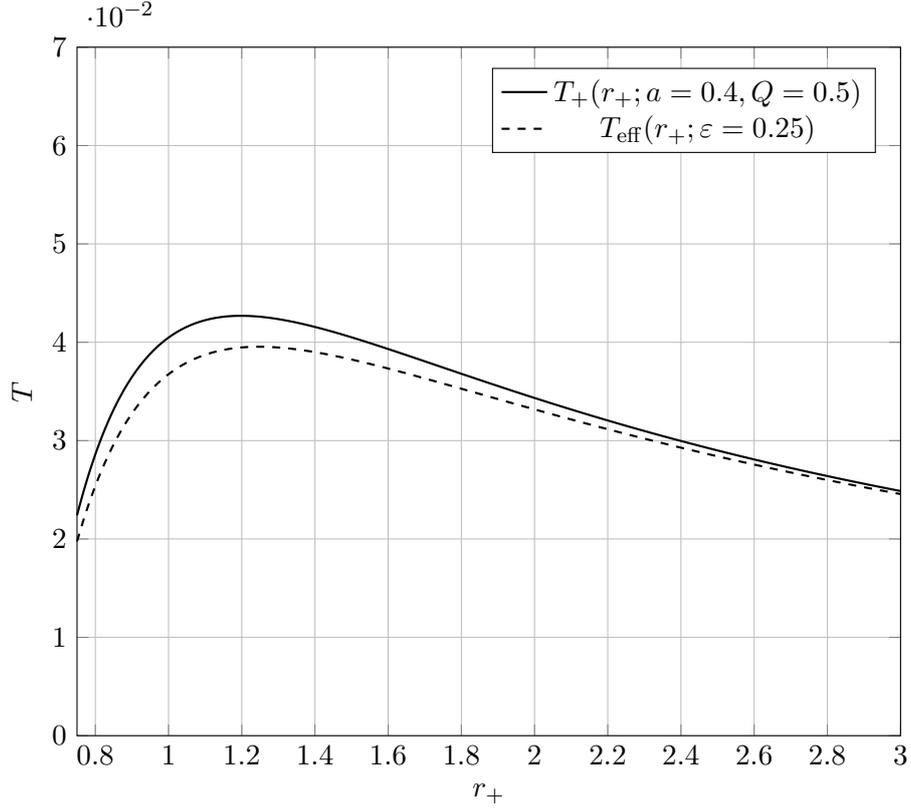
\begin{figure}[htbp]
\centering
\begin{tikzpicture}
\begin{axis}[
    width=0.78\textwidth,
    xlabel={$r_+$},
    ylabel={$T$},
    xmin=0.75, xmax=3,
    ymin=0, ymax=0.07,
    legend pos=north east,
    grid=both,
]
\addplot[smooth, thick, domain=0.75:3, samples=300] {((x - (0.5^2+0.4^2)/x)/(4*pi*(x^2+0.4^2)))};
\addlegendentry{$T_+(r_+;a=0.4,Q=0.5)$}
\addplot[smooth, thick, dashed, domain=0.75:3, samples=300] {((x - (0.5^2+0.4^2)/x)/(4*pi*(x^2+0.4^2)))*(1-0.25*exp(-x))};
\addlegendentry{$T_{\mathrm{eff}}(r_+;\varepsilon=0.25)$}
\end{axis}
\end{tikzpicture}
\caption{Outer-horizon Kerr--Newman temperature and a phenomenological non-equilibrium deformation. The dashed curve illustrates the suppression of the effective horizon temperature by weak dissipative dressing.}
\label{fig:temperature}
\end{figure}

\subsection{Irreversible entropy production}

Figure~\ref{fig:sirr} visualizes the quadratic entropy-production law \eqref{eq:quadraticSirr}.

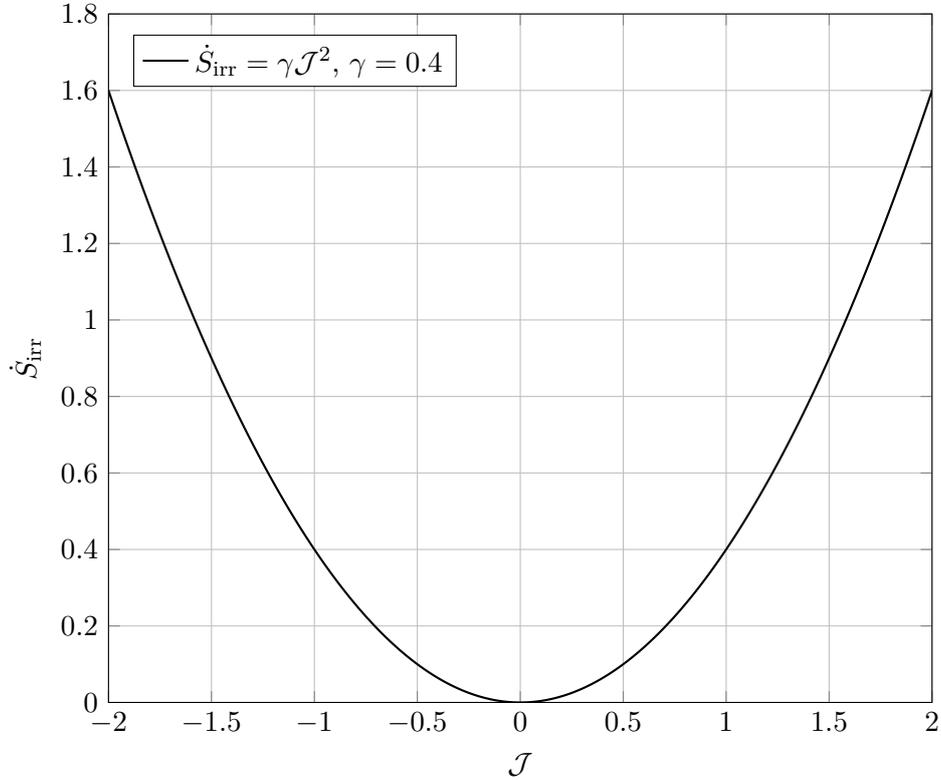
\begin{figure}[htbp]
\centering
\begin{tikzpicture}
\begin{axis}[
    width=0.78\textwidth,
    xlabel={$\mathcal{J}$},
    ylabel={$\dot S_{\mathrm{irr}}$},
    xmin=-2, xmax=2,
    ymin=0, ymax=1.8,
    legend pos=north west,
    grid=both,
]
\addplot[smooth, thick, domain=-2:2, samples=200] {0.4*x^2};
\addlegendentry{$\dot S_{\mathrm{irr}}=\gamma\mathcal{J}^2$, $\gamma=0.4$}
\end{axis}
\end{tikzpicture}
\caption{Quadratic irreversible entropy production as a function of a dimensionless flux variable \(\mathcal{J}\). The positivity of the curve encodes compatibility with the generalized second law.}
\label{fig:sirr}
\end{figure}

\section{From \texorpdfstring{$f(R)$}{f(R)} to a Gravity-Thermodynamized Model}

\subsection{Why one needs to upgrade from ``field equations'' to ``thermodynamization''}

The entropy-functional perspective indicates that if one requires the total entropy to be stationary, in an appropriate sense, for every local Rindler horizon, then the field equations may be regarded as the ``equation of state'' of spacetime \cite{Jacobson1995,PadmanabhanParanjape2007,Padmanabhan2010}. In the Einstein case, this idea is compatible with the local Clausius relation; however, in theories containing curvature corrections, a purely equilibrium relation is generally no longer sufficient, and one must introduce an additional entropy production term \cite{Jacobson1995,Padmanabhan2010}. For this reason, \(f(R)\) gravity is not as simple as mechanically replacing the area entropy \(A\) by \(FA\); rather, it naturally gives rise to a \emph{non-equilibrium horizon thermodynamics}. This line of thought is also closely aligned with the uploaded singular-ensemble framework \cite{ChenSingular2026}.

To systematize this point, we introduce the following definition.

\begin{definition}[Gravity-thermodynamized model]
Given a metric \(f(R)\) theory and a class of local or quasi-local horizons \(\cH\), we say that its \textbf{thermodynamization of gravity} is realized if:
\begin{enumerate}[label=(\roman*)]
\item the Wald entropy
\begin{equation}
\SH[\Sigma]=\frac{1}{4G}\int_{\Sigma}\Ff\,\dd A
\label{eq:WaldEntropy}
\end{equation}
is taken as the geometric entropy \cite{Wald1993};

\item the horizon heat flux \(\delta Q\), local temperature \(T\), reversible work terms, and internal entropy production \(\diS\) together form the complete entropy balance:
\begin{equation}
\dd S = \frac{\delta Q}{T}+\diS;
\label{eq:general_clausius}
\end{equation}

\item by redefining effective dark components, certain non-equilibrium descriptions are allowed to be rewritten into equilibrium ones;

\item \(\Ff\) itself is regarded as an internal thermodynamic variable, whose spatiotemporal variation drives irreversible effects.
\end{enumerate}
\end{definition}

\subsection{Local Rindler horizons, heat flux, and the variation of Wald entropy}

Consider a local Rindler horizon near a point \(p\), and let its null generator be
\(k^\mu=(\partial/\partial\lambda)^\mu\), where \(\lambda\) is an affine parameter. The near-horizon boost/Killing vector is taken as
\begin{equation}
\chi^\mu = -\kappa\lambda\,k^\mu,
\label{eq:boost_vector}
\end{equation}
where \(\kappa\) is the local surface gravity. The corresponding Unruh temperature is
\begin{equation}
T=\frac{\kappa}{2\pi}.
\label{eq:unruh_T}
\end{equation}
This local temperature assignment follows the standard Unruh effect and its horizon-thermodynamic interpretation \cite{Unruh1976,Jacobson1995}.

The matter heat flux crossing the horizon is defined by
\begin{equation}
\delta Q = \int_{\cH} T^{(m)}_{\mu\nu}\chi^\mu\,\dd\Sigma^\nu
= -\kappa\int_{\cH} \lambda\,T^{(m)}_{\mu\nu}k^\mu k^\nu\,\dd\lambda\,\dd A.
\label{eq:heat_flux}
\end{equation}

On the other hand, the variation of the Wald entropy is not merely an area variation, because \(\Ff\) itself also varies:
\begin{equation}
\delta \SH = \frac{1}{4G}\int_{\Sigma}\left(\Ff\,\delta \dd A + \dd A\,\delta \Ff\right).
\label{eq:deltaSH}
\end{equation}
Expressed differentially along the horizon generator, one obtains
\begin{equation}
\frac{\dd \SH}{\dd\lambda}
=\frac{1}{4G}\int_{\Sigma_\lambda}\left(\Ff\,\theta + \dot{\Ff}\right)\dd A,
\label{eq:dS_local}
\end{equation}
where
\begin{equation}
\theta\equiv \nabla_\mu k^\mu, \qquad
\dot{\Ff}\equiv k^\mu\nabla_\mu \Ff.
\label{eq:thetaFdot}
\end{equation}

To ensure that the first-order entropy variation vanishes on the local equilibrium cross section, one usually imposes
\begin{equation}
\left.(\Ff\theta+\dot{\Ff})\right|_{p}=0.
\label{eq:eq_slice}
\end{equation}
This condition is stronger than the Einstein-case requirement \(\theta|_p=0\): it demands that the area expansion and the variation of the scalar degree of freedom compensate each other on the equilibrium slice.

\subsection{Thermodynamic closure: reversible and irreversible parts}

Using the Raychaudhuri equation,
\begin{equation}
\dot\theta = -\frac{1}{2}\theta^2 - \sigma_{\mu\nu}\sigma^{\mu\nu} - R_{\mu\nu}k^\mu k^\nu,
\label{eq:Raychaudhuri}
\end{equation}
together with the field equation, the second-order variation of \eqref{eq:dS_local} can be organized into a ``heat-flux term + purely geometric term'' \cite{Jacobson1995,Padmanabhan2010}. In the \(f(R)\) case, the purely geometric contribution generally cannot be fully recast into \(\delta Q/T\), and one is naturally led to the non-equilibrium entropy balance
\begin{equation}
\dd \SH = \frac{\delta Q}{T} + \diS_{\mathrm{geom}}.
\label{eq:non_eq_balance}
\end{equation}

To elevate this into a genuine thermodynamic model, rather than merely a geometric rewriting of the equations, we further introduce a \emph{positive-definite constitutive closure}:
\begin{equation}
\diS_{\mathrm{irr}}
= \frac{1}{T}\int_{\cH}\dd\lambda\,\dd A
\left[
\zeta\,\Theta_{\mathrm{eff}}^2
+2\eta\,\sigma_{\mu\nu}\sigma^{\mu\nu}
+\chi\,\frac{(k^\mu\nabla_\mu \Ff)^2}{\Ff}
\right],
\qquad
\zeta,\eta,\chi\ge 0,
\label{eq:constitutive_entropy}
\end{equation}
where
\begin{equation}
\Theta_{\mathrm{eff}}\equiv \Ff\theta+\dot{\Ff}.
\label{eq:Thetaeff}
\end{equation}

\begin{remark}
Equation \eqref{eq:constitutive_entropy} is a new constitutive definition proposed here for the gravity-thermodynamized model; it is not explicitly stated in the uploaded preprints. Its physical meaning is:
\begin{itemize}
\item \(\sigma_{\mu\nu}\sigma^{\mu\nu}\) corresponds to shear dissipation;
\item \(\Theta_{\mathrm{eff}}^2\) corresponds to bulk-type or expansion-type irreversible dissipation;
\item \((k\cdot\nabla \Ff)^2/\Ff\) corresponds to internal entropy production induced by the scalaron flux.
\end{itemize}
The advantage of this form is that the second law,
\begin{equation}
\dtotS = \dmS + \dd \SH + \diS_{\mathrm{irr}} \ge 0,
\label{eq:secondlaw_total}
\end{equation}
is manifestly satisfied whenever the constitutive coefficients are non-negative.
\end{remark}

\subsection{Thermodynamic variational principle}

To unify thermodynamization with a variational principle, we define the Euclideanized Massieu-type functional as
\begin{equation}
\cM_{\mathrm{GT}}[g,\cH;\beta]
\equiv
\SH - \beta U_{\mathrm{quasi}} - \beta\,\cW_{\mathrm{irr}},
\label{eq:Massieu}
\end{equation}
where \(\beta=1/T\), \(U_{\mathrm{quasi}}\) is the energy function associated with the chosen quasi-local horizon, and \(\cW_{\mathrm{irr}}\) is the irreversible dissipation functional. Fixing the boundary temperature and external parameters, one imposes
\begin{equation}
\delta \cM_{\mathrm{GT}}=0,
\label{eq:massieu_variation}
\end{equation}
which yields
\begin{equation}
\delta U_{\mathrm{quasi}} = T\,\delta \SH - \delta \cW_{\mathrm{irr}} + \sum_i Y_i\,\delta X_i.
\label{eq:thermo_var_eq}
\end{equation}
Equation \eqref{eq:thermo_var_eq} is what we call the variational principle of thermodynamization of gravity: the field equations guarantee geometric integrability, the Wald entropy guarantees the definition of geometric entropy, and \(\cW_{\mathrm{irr}}\) controls the irreversible corrections away from equilibrium \cite{Wald1993,Padmanabhan2010,ChenSingular2026}.

\subsection{A rigorous reformulation of the residue idea}

The uploaded preprints attempt to relate poles, residues, and horizon thermodynamics through analytically continued functions \cite{ChenSingular2026,ChenTopological2025}. To avoid ambiguity in the physical interpretation of \(\exp(S)\), we suggest implementing the residue idea through the simple pole of the Euclideanized lapse function. For a static, spherically symmetric metric,
\begin{equation}
\dd s^2 = -B(r)\,\dd t^2 + \frac{\dd r^2}{B(r)} + r^2\dd\Omega_2^2,
\label{eq:sss_metric_simple}
\end{equation}
if \(r=r_h\) is a simple horizon, then
\begin{equation}
B(r)=B'(r_h)(r-r_h)+\cO((r-r_h)^2).
\label{eq:Bexpand}
\end{equation}
The absence of a conical singularity in Euclidean time requires
\begin{equation}
\beta_h = \frac{4\pi}{B'(r_h)}.
\label{eq:beta_simple}
\end{equation}
Meanwhile, the simple-pole residue immediately gives
\begin{equation}
\beta_h = 4\pi\,\Res_{r=r_h}\!\left(\frac{1}{B(r)}\right)
= \frac{2}{\ii}\oint_{\Gamma_h}\frac{\dd r}{B(r)}.
\label{eq:residue_temperature}
\end{equation}
That is, \textbf{the temperature is determined by the residue of the simple pole, the entropy is determined by the Wald--Noether charge, and the two are then coupled through the first law} \cite{Wald1993,GibbonsHawking1977,York1986,ChenSingular2026}. This is precisely a rigorous reformulation of the residue-based idea developed in the uploaded works.

\section{Static Spherically Symmetric Black Hole Example}

\subsection{Geometric preliminaries and horizon quantities}

For clarity, in this section we adopt the Schwarzschild gauge
\begin{equation}
\dd s^2 = -B(r)\,\dd t^2 + \frac{\dd r^2}{B(r)} + r^2\dd\Omega_2^2,
\qquad
B(r_h)=0,
\qquad
B'(r_h)\neq 0.
\label{eq:sss_metric}
\end{equation}
The corresponding Ricci scalar is
\begin{equation}
R(r) = -B''(r)-\frac{4B'(r)}{r}-\frac{2(B(r)-1)}{r^2}.
\label{eq:sss_R}
\end{equation}
The horizon temperature, entropy, and volume are defined by
\begin{equation}
\TempH = \frac{B'(r_h)}{4\pi}, \qquad
\SH = \frac{\Ahor\,\Ff_h}{4G}=\frac{\pi r_h^2\Ff_h}{G}, \qquad
\VH = \frac{4\pi r_h^3}{3},
\label{eq:sss_thermo_quantities}
\end{equation}
where \(\Ff_h\equiv \Ff(r_h)\). The entropy expression is the \(f(R)\) specialization of the Wald formula \cite{Wald1993}.

\subsection{The horizon first law derived from the field equations}

Writing the \(r{}_{r}\) component in mixed-index form, one finds
\begin{equation}
\Ff\left(\frac{B'}{r}+\frac{B-1}{r^2}\right)
-\frac{1}{2}B'\Ff'
-\frac{2B}{r}\Ff'
-\frac{1}{2}(R\Ff-f)
=
8\pi G\,P_r,
\label{eq:rr_component}
\end{equation}
where \(P_r\equiv T^r{}_{r}\). At the horizon \(r=r_h\), since \(B(r_h)=0\), this reduces to
\begin{equation}
8\pi G\,P_h
=
\Ff_h\left(\frac{B'_h}{r_h}-\frac{1}{r_h^2}\right)
-\frac{1}{2}B'_h\Ff'_h
-\frac{1}{2}(R_h\Ff_h-f_h).
\label{eq:rr_horizon}
\end{equation}

\begin{proposition}[Horizon first law for static spherically symmetric \(f(R)\) gravity]
Multiplying \eqref{eq:rr_horizon} by \(\dd\VH=4\pi r_h^2\dd r_h\), and introducing
\begin{equation}
\dd \EH \equiv \frac{1}{2G}\left[\Ff_h+\frac{r_h^2}{2}(R_h\Ff_h-f_h)\right]\dd r_h,
\label{eq:dEH_def}
\end{equation}
the field equation can be rewritten as
\begin{equation}
\dd \EH = \TempH\,\dd \SH - P_h\,\dd \VH + \TempH\,\diS_h,
\label{eq:sss_firstlaw}
\end{equation}
where the geometric additional entropy differential is
\begin{equation}
\TempH\,\diS_h = -\frac{r_h^2 B'_h\Ff'_h}{2G}\,\dd r_h.
\label{eq:sss_diS}
\end{equation}
\end{proposition}

\begin{proof}
First, from \eqref{eq:sss_thermo_quantities} one computes
\begin{equation}
\dd \SH = \frac{\pi}{G}\left(2r_h\Ff_h+r_h^2\Ff'_h\right)\dd r_h,
\label{eq:dSHcalc}
\end{equation}
hence
\begin{equation}
\TempH\,\dd \SH
=
\frac{1}{4G}\left(2r_h\Ff_hB'_h+r_h^2B'_h\Ff'_h\right)\dd r_h.
\label{eq:TdS_calc}
\end{equation}
Next, using \eqref{eq:rr_horizon} to eliminate \(P_h\dd V_h\), one obtains
\begin{equation}
P_h\,\dd V_h
=
\frac{1}{2G}
\left[
\Ff_h(B'_hr_h-1)
-\frac{1}{2}r_h^2B'_h\Ff'_h
-\frac{1}{2}r_h^2(R_h\Ff_h-f_h)
\right]\dd r_h.
\label{eq:PdV_calc}
\end{equation}
Substituting \eqref{eq:TdS_calc} and \eqref{eq:PdV_calc} into the desired expression and rearranging terms yields \eqref{eq:sss_firstlaw}, where \(\dd E_h\) and \(\diS_h\) are respectively defined by \eqref{eq:dEH_def} and \eqref{eq:sss_diS}.
\end{proof}

\begin{remark}
If \(f(R)=R-2\Lambda\), then \(\Ff=1\) and \(\Ff'_h=0\), so \eqref{eq:sss_diS} vanishes, and \eqref{eq:sss_firstlaw} reduces to the equilibrium first law at the horizon in Einstein gravity. In other words, \(\Ff'_h\neq 0\) is the most direct signal that \(f(R)\) induces a non-equilibrium contribution \cite{Jacobson1995,Padmanabhan2010}.
\end{remark}

\subsection{Thermodynamical interpretation}

The quantity \(\diS_h\) in \eqref{eq:sss_firstlaw} admits two readings:
\begin{enumerate}[label=(\alph*)]
\item it may be regarded as a \textbf{geometric additional term} induced by higher-derivative geometry;

\item in a coarse-grained description, it may be mapped onto the \textbf{irreversible entropy production} in \eqref{eq:constitutive_entropy}.
\end{enumerate}
The second interpretation is more suitable for the language of thermodynamization of gravity: in that case, \(\Ff\) is no longer merely the derivative of the Lagrangian, but becomes an internal variable in spacetime thermodynamics. Its radial gradient \(\Ff'_h\) characterizes the scalar thermodynamic inhomogeneity near the horizon.

\subsection{FRW Cosmological Example: Rewritten Friedmann Equations}

Consider the FRW metric
\begin{equation}
\dd s^2 = -\dd t^2 + a(t)^2\left[\frac{\dd r^2}{1-kr^2}+r^2\dd\Omega_2^2\right],
\label{eq:FRW_metric}
\end{equation}
where \(H=\dot a/a\), \(k=0,\pm1\), and the Ricci scalar is
\begin{equation}
R = 6\left(2H^2+\dot H + \frac{k}{a^2}\right).
\label{eq:FRW_R}
\end{equation}
For a perfect fluid \(T^\mu{}_{\nu}=\mathrm{diag}(-\rho,p,p,p)\), the field equations give
\begin{align}
3\Ff\left(H^2+\frac{k}{a^2}\right)
&=8\pi G\rho + \frac{1}{2}(\Ff R-f)-3H\dot{\Ff},
\label{eq:Friedmann1_fR}\\
-2\Ff\left(\dot H-\frac{k}{a^2}\right)
&=8\pi G(\rho+p)+\ddot{\Ff}-H\dot{\Ff}.
\label{eq:Friedmann2_fR}
\end{align}
Meanwhile, matter conservation reads
\begin{equation}
\dot\rho+3H(\rho+p)=0.
\label{eq:matter_conservation}
\end{equation}

\subsection{Thermodynamic quantities at the apparent horizon}

The apparent-horizon radius in FRW spacetime is defined by
\begin{equation}
\rA = \frac{1}{\sqrt{H^2+k/a^2}},
\label{eq:rA_def}
\end{equation}
which satisfies
\begin{equation}
\dot\rA = -H\rA^3\left(\dot H-\frac{k}{a^2}\right).
\label{eq:rA_dot}
\end{equation}
The area, entropy, volume, and work density at the apparent horizon are defined by
\begin{equation}
A_A = 4\pi \rA^2,\qquad
\SA = \frac{A_A\Ff}{4G}=\frac{\pi \rA^2\Ff}{G}, \qquad
V_A=\frac{4\pi\rA^3}{3}, \qquad
W=\frac{\rho-p}{2}.
\label{eq:FRW_thermo_quantities}
\end{equation}
The apparent-horizon temperature is taken to be
\begin{equation}
\TA = \frac{|\kappa_A|}{2\pi}, \qquad
\kappa_A = -\frac{1}{\rA}\left(1-\frac{\dot\rA}{2H\rA}\right).
\label{eq:TA_def}
\end{equation}

\subsection{Non-equilibrium form of the first law in FRW thermodynamics}

Define the matter energy inside the horizon by
\begin{equation}
\EA = \rho V_A.
\label{eq:EA_def}
\end{equation}
Using \eqref{eq:matter_conservation}, one obtains
\begin{equation}
\dd \EA = 4\pi\rA^2\rho\,\dd\rA - 4\pi\rA^3H(\rho+p)\,\dd t.
\label{eq:dEA}
\end{equation}
Combining \eqref{eq:Friedmann1_fR}--\eqref{eq:Friedmann2_fR}, \eqref{eq:rA_dot}, \eqref{eq:FRW_thermo_quantities}, and \eqref{eq:dEA}, the field equations can be rewritten as
\begin{equation}
\dd \EA = \TA\,\dd \SA + W\,\dd V_A + \TA\,\diS_A,
\label{eq:FRW_firstlaw_noneq}
\end{equation}
where
\begin{equation}
\TA\,\diS_A \equiv \dd \EA - \TA\,\dd \SA - W\,\dd V_A.
\label{eq:FRW_diS_residual}
\end{equation}

\begin{remark}
Equation \eqref{eq:FRW_diS_residual} \emph{defines} the additional entropy term as the residual geometric contribution after rewriting the first law. The advantage of this construction is that it is fully rigorous and does not presuppose any specific constitutive model. If one further adopts an equilibrium redefinition scheme, in which the geometric corrections are absorbed into an effective dark-component energy-momentum tensor, then \eqref{eq:FRW_firstlaw_noneq} can be rewritten into an equilibrium form without \(\diS_A\). This distinction between equilibrium and non-equilibrium horizon thermodynamics is consistent with the general emergent-gravity logic \cite{Jacobson1995,Padmanabhan2010}.
\end{remark}

\subsection{The second law in FRW and the dual equilibrium/non-equilibrium descriptions}

If the temperature of the cosmic matter inside the horizon is assumed to be equal to the temperature of the apparent horizon, and the geometric additional term is coarse-grained into a positive-definite constitutive contribution, then the total entropy production rate may be written as
\begin{equation}
\frac{\dd S_{\mathrm{tot}}}{\dd t}
=
\frac{\dd}{\dd t}\left(S_m+\SA\right) + \frac{\dd S_{i,A}}{\dd t}
\ge 0.
\label{eq:FRW_secondlaw}
\end{equation}
On the other hand, one may also adopt an equilibrium description by rewriting the horizon entropy into an area-type entropy and absorbing the non-equilibrium contribution into a redefined dark fluid. Both descriptions are established in the literature, but the origins of entropy production are not exactly the same. It is therefore conceptually preferable to distinguish between the local Rindler non-equilibrium term and the FRW horizon rewriting term.

\section{Discussion}

The present construction integrates the main ideas of the uploaded works into a single non-equilibrium perspective \cite{ChenTopological2025,ChenSingular2026}. The entropy-functional criterion provides a background-selection rule, the residue formalism makes the thermodynamic role of the horizon singularity explicit, and the topological index captures the net orientation of thermodynamic branches. Several observations follow.

First, the residue calculus is robust because it depends only on the local singularity class of the lapse function \cite{GibbonsHawking1977,York1986,ChenSingular2026}. This makes the formalism relatively insensitive to many details of the bulk geometry away from the horizon. Such insensitivity is useful when one aims to isolate universal aspects of black-hole thermodynamics.

Second, the non-equilibrium extension is deliberately quasi-stationary. A fully time-dependent, strongly non-equilibrium black hole requires the full dynamical metric and probably a real-time Schwinger--Keldysh-type formulation. Our framework instead targets slowly driven states for which an instantaneous horizon and an instantaneous contour are still meaningful. In that regime, Eqs.~\eqref{eq:lnZneq} and \eqref{eq:entropybalance} provide a transparent extension of equilibrium thermodynamics.

Third, the topological classification suggests a protected structure. As long as the number and orientation of thermodynamic branches do not change, the index \(W\) is unchanged even in the presence of weak \(f(R)\) corrections or weak dissipative driving \cite{YerraBhamidipati2022,ChenTopological2025}. This echoes the intuition that topological data are more robust than local response coefficients.

Finally, the explicit function plots should be interpreted correctly. They are not substitutes for a numerical solution of the non-equilibrium \(f(R)\) field equations. Rather, they serve as analytically controlled visualizations of the formalism: the free-energy model displays branch lifting, the temperature plot shows dissipative dressing of the horizon response, and the entropy-production curve emphasizes the irreversible sector. In a broader context, this remains compatible both with microscopic entropy approaches \cite{StromingerVafa1996,Strominger1998} and with thermodynamic-geometry interpretations of black-hole microstructure \cite{Ruppeiner1995,WeiLiuMann2019}.

\subsection{Conclusion}

We have formulated the non-equilibrium physics of thermodynamicized black holes by synthesizing the methods of entropy-functional background selection, contour-residue thermodynamics, and topological horizon classification. The central results are as follows:
\begin{enumerate}[label=(\arabic*)]
    \item the equilibrium horizon temperature remains encoded by a simple-pole residue, \(\beta_{h}=4\pi\,\Res(1/f)\);
    \item a quasi-stationary non-equilibrium state is described by a generalized singular action containing reversible work terms and an irreversible contribution \(\Pi_{h}\);
    \item the entropy balance law acquires the natural form of a reversible first-law sector plus a positive entropy-production term \(\dot S_{\mathrm{irr}}\);
    \item for Kerr--Newman-type black holes in constant-curvature \(f(R)\) gravity, the entropy continues to be weighted by \(f'(R_{0})\), while the non-extremal topological class remains \(W=0\) unless the horizon structure itself changes.
\end{enumerate}
The framework is suitable for further generalization to multi-horizon de Sitter black holes, rotating AdS backgrounds, or transport-theoretic derivations of the dissipative coefficients. It may also provide a useful bridge between residue-based horizon thermodynamics and more microscopic descriptions of black-hole non-equilibrium dynamics.

\end{document}